\documentclass[12pt]{article}

\usepackage[margin=2.5cm]{geometry}

\usepackage{amsmath, amssymb, latexsym, amsthm}

\usepackage{graphicx}

\newtheorem{theorem}{Theorem}[section]

\newtheorem{lemma}[theorem]{Lemma}

\newtheorem*{theorem*}{Theorem}
\newtheorem*{conjecture*}{Conjecture}
\newtheorem*{corollary*}{Corollary}
\newtheorem*{proposition*}{Proposition}

\theoremstyle{definition}

\newenvironment{acknowledgement}[1][Acknowledgement]{\begin{trivlist}
\item[\hskip \labelsep {\bfseries #1}]}{\end{trivlist}}

\theoremstyle{remark}

\begin{document}

\title{On interval edge-colorings of planar graphs}

\author{Arsen Hambardzumyan, Levon Muradyan}

\author{
{\sl Arsen Hambardzumyan}\thanks{{\it E-mail address:} 
hambardzumyanarsen99@gmail.com}\\ 
Department of Informatics \\ and Applied Mathematics,\\
Yerevan State University \\ 0025, Armenia
\and
{\sl Levon Muradyan}\thanks{{\it E-mail address:} 
levonmuradyanlevon@gmail.com}\\ 
Department of Informatics \\ and Applied Mathematics,\\
Yerevan State University \\ 0025, Armenia
}

\maketitle


\begin{abstract}
An edge-coloring of a graph $G$ with colors $1,\ldots,t$ is called an \emph{interval $t$-coloring} if all colors are used and the colors of edges incident to each vertex of $G$ are distinct and form an interval of integers.
In 1990, Kamalian proved that if a graph $G$ with at least one edge has an interval $t$-coloring, then $t\leq 2|V(G)|-3$. In 2002, Axenovich improved this upper bound for planar graphs: if a planar graph $G$ admits an interval $t$-coloring, then $t\leq \frac{11}{6}|V(G)|$. In the same paper Axenovich suggested a conjecture that if a planar graph $G$ has an interval $t$-coloring, then $t\leq \frac{3}{2}|V(G)|$. In this paper we confirm the conjecture by showing that if a planar graph $G$ admits an interval $t$-coloring, then $t\leq \frac{3|V(G)|-4}{2}$. We also prove that if an outerplanar graph $G$ has an interval $t$-coloring, then $t\leq |V(G)|-1$. Moreover, all these upper bounds are sharp.        

\end{abstract}
	
\bigskip
\section{Introduction}

We use \cite{West} for terminology and notation not defined here. We consider graphs that are finite, undirected, and have no loops or multiple edges. Let $V(G)$ and $E(G)$ denote the sets of vertices and edges of a graph $G$, respectively. The degree of a vertex $v\in V(G)$ is denoted by $d_{G}(v)$, the maximum degree of $G$ by $\Delta(G)$, the diameter of $G$ by $\mathrm{diam}(G)$ and the chromatic index of $G$ by $\chi^{\prime}(G)$. A proper edge-coloring of a graph $G$ is a mapping $\alpha: E(G)\rightarrow \mathbb{N}$ such that $\alpha(e)\not=\alpha(e')$ for every pair of adjacent edges $e$ and $e'$ in $G$. 

An \emph{interval $t$-coloring} of a graph $G$ is a proper edge-coloring $\alpha$ of $G$ with colors $1,\ldots,t$ such that all colors are used and for each $v\in V(G)$, the set of colors of the edges incident to $v$ is an interval of integers. A graph $G$ is \emph{interval colorable} if there is an integer $t\geq 1$ for which $G$ has an interval $t$-coloring. The set of all interval colorable graphs is denoted by $\mathfrak{N}$. For a graph $G\in \mathfrak{N}$, the maximum value of $t$ for which $G$ has an interval $t$-coloring is denoted by $W(G)$. The notion of interval colorings was introduced by Asratian and Kamalian \cite{AsrKam} (available in English as \cite{AsrKamJCTB}) in 1987 and was motivated by the problem of finding compact school timetables, that is, timetables such that the lectures of each teacher and each class are scheduled at consecutive periods. This problem corresponds to the problem of finding an interval edge-coloring of a bipartite multigraph. 
In \cite{AsrKam,AsrKamJCTB}, Asratian and Kamalian noted that if $G$ is interval
colorable, then $\chi^{\prime }\left(G\right)=\Delta(G)$. Moreover,
they also showed that it is an $NP$-complete problem to determine whether an $r$-regular ($r\geq 3$) graph has an interval coloring or not. Asratian and Kamalian also proved \cite{AsrKam,AsrKamJCTB} that if a triangle-free graph $G$ has an interval $t$-coloring, then $t\leq \left\vert V(G)\right\vert -1$. In \cite{Kampreprint}, Kamalian investigated interval colorings of complete bipartite graphs and trees. In particular, he proved that the complete bipartite graph
$K_{m,n}$ has an interval $t$-coloring if and only if
$m+n-\gcd(m,n)\leq t\leq m+n-1$, where $\gcd(m,n)$ is the greatest
common divisor of $m$ and $n$. In \cite{Petrosyan,PetrosyanKhachatrianTananyan}, Petrosyan, Khachatrian and Tananyan investigated interval colorings of complete graphs and
$n$-dimensional cubes. In particular, they proved that the $n$-dimensional cube $Q_{n}$ has an interval $t$-coloring if and only if $n\leq t\leq \frac{n\left(n+1\right)}{2}$. Generally, it is an $NP$-complete problem to determine whether a bipartite graph has an interval coloring \cite{Seva}.
In fact, for every positive integer $\Delta\geq 11$, there exists a bipartite graph with maximum degree $\Delta$ that has no interval coloring \cite{PetrosHrant}. However, some classes of graphs have been proved to admit interval colorings; it is known, for example, that trees, regular and complete bipartite graphs \cite{AsrKam,Hansen,Kampreprint}, subcubic graphs with $\chi^{\prime }\left(G\right)=\Delta(G)$ \cite{ArmenCarlPetros}, doubly convex bipartite graphs \cite{AsrDenHag,KamDiss}, grids \cite{GiaroKubale1}, outerplanar bipartite graphs \cite{GiaroKubale2}, $(2,b)$-biregular graphs \cite{Hansen,HansonLotenToft,KamMir} and $(3,6)$-biregular graphs \cite{CarlJToft} have interval colorings, where an \emph{$(a,b)$-biregular} graph is a bipartite graph where the vertices in one part all have degree $a$ and the vertices in the other part all have degree $b$.

General upper bounds on the number of colors in interval edge-colorings of graphs were obtained in \cite{KamDiss,GiaroKubaleMalaf}. In particular, Kamalian \cite{KamDiss} proved that if $G$ is a simple graph with at least one edge and $G\in \mathfrak{N}$, then $W(G)\leq 2|V(G)|-3$. Moreover, this upper bound is sharp for $K_{2}$. In 2001, Giaro, Kubale and Ma\l afiejski \cite{GiaroKubaleMalaf}  slightly improved the upper bound by showing that if $G$ is a simple graph with at least $3$ vertices and $G\in \mathfrak{N}$, then $W(G)\leq 2|V(G)|-4$. On the other hand, in \cite{Petrosyan} it was proved that for any $\varepsilon>0$, there exists a graph $G$ such that $G\in \mathfrak{N}$ and $W(G)\geq (2-\varepsilon)|V(G)|$. In the case of planar graphs general upper bounds on the parameter $W(G)$ were improved by Axenovich \cite{Axen}. In particular, she proved that if $G$ is a planar graph and $G\in \mathfrak{N}$, then $W(G)\leq \frac{11}{6}|V(G)|$, and conjectured that this upper bound can be improved to $W(G)\leq \frac{3}{2}|V(G)|$. In \cite{AsrKamJCTB}, Asratian
and Kamalian proved that if $G$ is connected and $G\in\mathfrak{N}$, then 
$W(G)\leq\left(\mathrm{diam}(G)+1\right)\left(\Delta(G)-1\right) +1$. 
They also proved that if $G$ is connected bipartite and $G\in\mathfrak{N}$, then this bound can be improved to $W(G)\leq \mathrm{diam}(G)\left(\Delta(G)-1\right) +1$. Kamalian and Petrosyan \cite{KamalianPetrosyan} showed that these upper bounds cannot be significantly improved. Recently, Casselgren, Khachatrian and Petrosyan \cite{CarlHrantPetros} derived upper bounds on the number of colors in interval edge-colorings of multigraphs. In particular, they proved that if $G$ is a connected cubic multigraph and $G\in\mathfrak{N}$, then $W(G)\leq |V(G)|+1$ and this upper bound is sharp.    

In this paper we confirm Axenovich's conjecture on interval edge-colorings of planar graphs. We also prove that if $G$ is an outerplanar graph and $G\in\mathfrak{N}$, then  $W(G)\leq |V(G)|-1$ and this upper bound is sharp.

\section{Main Results}

Before we formulate and prove our main results, we need two simple properties on planar and outerplanar graphs and one new lemma. 

\begin{lemma}\cite{West}\label{LemmaPlanar} A planar graph $G$ of order $n\geq 2$ ($n\geq 3$) has at most $3n-5$ ($3n-6$) edges. 
\end{lemma}

\begin{lemma}\cite{West}\label{LemmaOuterplanar} An outerplanar graph $G$ of order $n$ ($n\geq 2$) has at most $2n-3$ edges.
\end{lemma}
	
We also need the following simple lemma.

\begin{lemma}\label{Lemma} If a graph $G$ has an interval $W(G)$-coloring $\alpha$ and the number of colors that are used only once in the coloring $\alpha$ is $k$, then $$W(G) \leq \frac{|E(G)| + k}{2}.$$
\end{lemma}
\begin{proof}

Since an interval $W(G)$-coloring $\alpha$ of $G$ has only $k$ colors that are used only once and the other $W(G) - k$ colors are used at least twice, we obtain
$$
|E(G)| \geq k + 2(W(G) - k), \text{~thus}
$$
$$
W(G) \leq \frac{|E(G)| + k}{2}.
$$
\end{proof}

We now able to prove our first main result. 

\begin{theorem}\label{Planar} If $G$ is a planar graph with at least two vertices and $G\in\mathfrak{N}$, then 
$$W(G) \leq \frac{3|V(G)|-4}{2}.$$
\end{theorem}
\begin{proof}
Let $\alpha$ be an interval $W(G)$-coloring of $G$ and $k$ be the number of colors that are used only once in the coloring $\alpha$. If $k \leq 1$, then, by Lemmas \ref{LemmaPlanar} and \ref{Lemma}, we have 
$$W(G) \leq \frac{|E(G)| + k}{2} \leq \frac{3|V(G)|-5+1}{2} = \frac{3|V(G)|-4}{2}.$$

So we may assume that $k\geq 2$.

Let us denote by $e_1,e_2,\ldots,e_k$ the edges of $G$ whose colors are used only once in the coloring $\alpha$. Let $e_i= u_iv_i$ and $c_i = \alpha(e_i)$ ($1 \leq i \leq k$). Without loss of generality we may assume that $c_1 < c_2 < \cdots < c_k$. Let us also define $c_0 = 1$ and $c_{k+1}=W(G)$. Clearly, $c_0\leq c_1$ and $c_k\leq c_{k+1}$.

For each $i$ ($1 \leq i \leq k+1$), let $C_i=\{e|~e\in E(G)\text{~and~}\alpha(e)\leq c_{i}\}$ and $G_i$ be the subgraph of the graph $G$ containing all the edges whose colors are smaller than or equal to $c_i$ ($C_i$) and their incident vertices. Also, for each $i$ ($1 \leq i \leq k+1$), let $C'_i=\{e|~e\in E(G)\text{~and~}c_{i-1}\leq \alpha(e)\leq c_{i}\}$ and $G'_i$ be the subgraph of the graph $G$ containing all the edges whose colors are larger than or equal to $c_{i-1}$ and smaller than or equal to $c_i$ ($C'_i$) and their incident vertices. Clearly, $G_1 = G'_1$ and $G = G_{k+1}$. Let us also note that since $G$ is a planar graph, for each $i$ ($1 \leq i \leq k+1$), subgraphs $G_i$ and $G'_i$ of $G$ are planar too.   

Let us prove that $V(G'_i) \cap V(G'_{i+1})=\{u_i, v_i\}$ for each $i$ ($1 \leq i \leq k$).

Clearly, by the definitions of $G'_i$ and $G'_{i+1}$, we have $\{u_i, v_i\}\subseteq V(G'_i)\cap V(G'_{i+1})$. Let us now assume that $V(G'_i) \cap V(G'_{i+1})$ contains the vertex $w$ that is neither $u_i$ nor $v_i$.
Since $w \in V(G'_i)$, by the definition of $G'_i$, there exists an edge $e'\in E(G)$ which is incident to $w$ such that $\alpha(e') \leq c_i$. Similarly, since $w \in V(G'_{i+1})$, by the definition of $G'_{i+1}$, there exists an edge $e''\in E(G)$ which is incident to $w$ such that $\alpha(e'') \geq c_i$. Since $\alpha$ is an interval coloring, there exists an edge in $G$ with color $c_i$ which is incident to the vertex $w$. As we mentioned before the color $c_i$ occurs in the coloring $\alpha$ only once and that color has only edge $e_i$ which is not incident to $w$, thus $w \notin V(G'_i) \cap V(G'_{i+1})$. Hence, $V(G'_i) \cap V(G'_{i+1})=\{u_i, v_i\}$ for each $i$ ($1 \leq i \leq k$). This implies that $C'_i \cap C'_{i+1} = \{e_i\}$ for each $i$ ($1 \leq i \leq k$).

Since for each $i$ ($2 \leq i \leq k + 1$), $V(G'_{i-1}) \cap V(G'_{i})=\{u_{i-1}, v_{i-1}\}$ and $C'_{i-1} \cap C'_{i} = \{e_{i-1}\}$, we obtain that the following two properties hold:
\begin{description}
    \item[(a)] for each $i$ ($2 \leq i \leq k + 1$), $|V(G'_i)| = |V(G_{i}) \setminus V(G_{i-1})| + 2$;
    \item[(b)] for each $i$ ($2 \leq i \leq k+1$), $C'_i = \left(C_{i} \setminus C_{i-1}\right) \cup \{e_{i-1}\}$ and hence, $|C_i \setminus C_{i-1}| = |C'_i| - 1$.    
\end{description}    

On the other hand, since for each $i$ ($1 \leq i \leq k+1$), subgraphs $G_i$ and $G'_i$ of $G$ are planar, by Lemma \ref{LemmaPlanar}, we obtain that the following two properties hold:
\begin{description}
    \item[(c)] $|C'_1| \leq 3|V(G'_1)| - 5$ and $|C'_{k+1}|\leq 3|V(G'_{k+1})| - 5$, because these planar subgraphs may contain only two vertices;
    \item[(d)] for each $i$ ($2 \leq i \leq k$), $|C'_i| \leq 3|V(G'_i)| - 6$, because  $|V(G'_i)| \geq 3$ ($G'_i$ contains at least two edges $e_{i-1}$ and $e_{i}$). 
\end{description}    
 
Now, using aforementioned properties (a)-(d) and taking into account Lemma \ref{Lemma}, we have

\begin{gather*}
    W(G) = W(G_{k+1}) \leq \frac{|C_{k+1}| + k}{2} = \frac{|C_1| + \sum_{i=2}^{k+1}{|C_{i} \setminus C_{i-1}|}+k}{2} = \\
    = \frac{|C'_1|   + \sum_{i=2}^{k+1}{(|C'_{i}|-1)}+k}{2} 
    = \frac{|C'_1| + |C'_{k+1}| + \sum_{i=2}^{k}{|C'_{i}|}-k+k}{2}\leq \\
    \leq \frac{3|V(G'_1)| - 5 + 3|V(G'_{k+1})| - 5 + \sum_{i=2}^{k}{(3|V(G'_{i})| - 6)}}{2} =\\
    =\frac{3|V(G'_1)| - 5 + 1 + \sum_{i=2}^{k+1}{(3|V(G'_{i})| - 6)}}{2} = \\
    = \frac{3|V(G_1)| - 4 + \sum_{i=2}^{k+1}{(3(|V(G_{i}) 
\setminus V(G_{i-1})| + 2) - 6)}}{2} = \\
    = \frac{3|V(G_1)| - 4 + 3\sum_{i=2}^{k+1}{|V(G_{i}) 
\setminus V(G_{i-1})|}}{2} = \frac{3|V(G_{k+1})| - 4}{2} = \frac{3|V(G)| - 4}{2}.
\end{gather*}
\end{proof}

In \cite{Axen}, Axenovich constructed a family of interval colorable planar graphs $G_{n}$ of order $n$ and $W(G_{n})=\frac{3n-4}{2}$. Such an example of the interval colorable planar graph $G_{8}$ is shown in Fig. \ref{fig}.

\begin{figure}[t]
\begin{center}
\includegraphics[width=11cm]{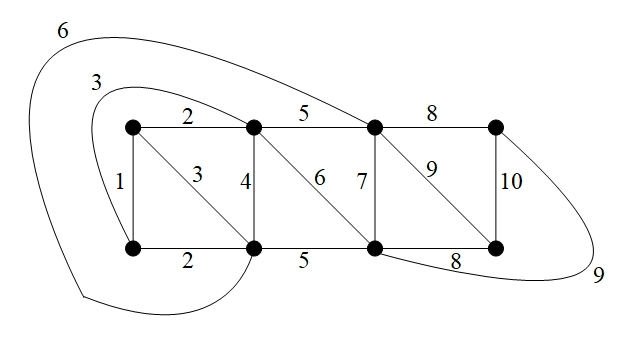}
\caption{A planar graph $G_{8}$ on $8$ vertices with an interval $10$-coloring.}\label{fig}
\end{center}
\end{figure}

Our next result concerns interval edge-colorings of outerplanar graphs. 

\begin{theorem} If $G$ is an outerplanar graph with at least two vertices and $G\in\mathfrak{N}$, then 
$$W(G) \leq |V(G)|-1.$$
\end{theorem}
\begin{proof}
As in the proof of Theorem \ref{Planar}, let $\alpha$ be an interval $W(G)$-coloring of $G$ and $k$ be the number of colors that are used only once in the coloring $\alpha$. If $k \leq 1$, then, by Lemmas \ref{LemmaOuterplanar} and \ref{Lemma}, we have 
$$W(G) \leq \frac{|E(G)| + k}{2} \leq \frac{2|V(G)|-3+1}{2} = |V(G)|-1.$$

So we may assume that $k\geq 2$.

Let us denote by $e_1,e_2,\ldots,e_k$ the edges of $G$ whose colors are used only once in the coloring $\alpha$. Let $e_i= u_iv_i$ and $c_i = \alpha(e_i)$ ($1 \leq i \leq k$). Without loss of generality we may assume that $c_1 < c_2 < \cdots < c_k$. Let us also define $c_0 = 1$ and $c_{k+1}=W(G)$. Clearly, $c_0\leq c_1$ and $c_k\leq c_{k+1}$.
 
For each $i$ ($1 \leq i \leq k+1$), let $C_i=\{e|~e\in E(G)\text{~and~}\alpha(e)\leq c_{i}\}$ and $G_i$ be the subgraph of the graph $G$ containing all the edges whose colors are smaller than or equal to $c_i$ ($C_i$) and their incident vertices. Also, for each $i$ ($1 \leq i \leq k+1$), let $C'_i=\{e|~e\in E(G)\text{~and~}c_{i-1}\leq \alpha(e)\leq c_{i}\}$ and $G'_i$ be the subgraph of the graph $G$ containing all the edges whose colors are larger than or equal to $c_{i-1}$ and smaller than or equal to $c_i$ ($C'_i$) and their incident vertices. Clearly, $G_1 = G'_1$ and $G = G_{k+1}$. Let us also note that since $G$ is an outerplanar graph, for each $i$ ($1 \leq i \leq k+1$), subgraphs $G_i$ and $G'_i$ of $G$ are outerplanar too.

Similarly, as in the proof of Theorem \ref{Planar}, it can be shown that $V(G'_i) \cap V(G'_{i+1})=\{u_i, v_i\}$ and $C'_i \cap C'_{i+1} = \{e_i\}$ for each $i$ ($1 \leq i \leq k$).

Since for each $i$ ($2 \leq i \leq k + 1$), $V(G'_{i-1}) \cap V(G'_{i})=\{u_{i-1}, v_{i-1}\}$ and $C'_{i-1} \cap C'_{i} = \{e_{i-1}\}$, we obtain that the following two properties hold:
\begin{description}
    \item[(1)] for each $i$ ($2 \leq i \leq k+1$), $|V(G'_i)| = |V(G_{i}) \setminus V(G_{i-1})| + 2$;
    \item[(2)] for each $i$ ($2 \leq i \leq k+1$), $C'_i = \left(C_{i} \setminus C_{i-1}\right) \cup \{e_{i-1}\}$ and hence, $|C_i \setminus C_{i-1}| = |C'_i| - 1$.    
\end{description}    

On the other hand, since for each $i$ ($1 \leq i \leq k+1$), subgraphs $G_i$ and $G'_i$ of $G$ are outerplanar, by Lemma \ref{LemmaOuterplanar}, we obtain that the following property holds:
\begin{description}
    \item[(3)] for each $i$ ($1 \leq i \leq k+1$), $|C'_i| \leq 2|V(G'_i)| - 3$, because $|V(G'_i)|\geq 2$.
 \end{description}   

Since an interval $W(G)$-coloring $\alpha$ of $G$ has only $k$ colors that are used only once and the other $W(G) - k$ colors are used at least twice, for each $i$ ($1 \leq i \leq k+1$), let us estimate the maximum number of distinct colors that are used in $G'_i$ except for unique colors $c_1,c_2,\ldots,c_k$:

\begin{description}
    \item[(4)] $G'_1$ contains only one edge which is colored with one of the $k$ unique colors, thus the maximum number of distinct colors that are used in $G'_1$ except for unique colors $c_1,c_2,\ldots,c_k$ is at most $\left\lfloor\frac{|C'_1| - 1}{2}\right\rfloor$;
    \item[(5)] for each $i$ ($2 \leq i \leq k$), $G'_i$ contains only two edges which are colored with two of the $k$ unique colors, thus the maximum number of colors that are used in $G'_i$ except for unique colors $c_1,c_2,\ldots,c_k$ is at most $\left\lfloor\frac{|C'_i| - 2}{2}\right\rfloor$;
    \item[(6)] $G'_{k+1}$ contains only one edge which is colored with one of the $k$ unique colors, thus the maximum number of colors that are used in $G'_{k + 1}$ except for unique colors $c_1,c_2,\ldots,c_k$ is at most $\left\lfloor\frac{|C'_{k+1}| - 1}{2}\right\rfloor$.
\end{description}
 
Now, using aforementioned properties (1)-(6), we have

\begin{gather*}
    W(G) \leq k + \left\lfloor\frac{|C'_1| - 1}{2}\right\rfloor + \left\lfloor\frac{|C'_{k+1}| - 1}{2}\right\rfloor + \sum_{i=2}^{k}{\left\lfloor\frac{|C'_{i}|-2}{2}\right\rfloor} \leq \\
    \leq k + \left\lfloor\frac{2|V(G'_1)| - 3 - 1}{2}\right\rfloor + \left\lfloor\frac{2|V(G'_{k+1})| - 3 - 1}{2}\right\rfloor + \sum_{i=2}^{k}{\left\lfloor\frac{2|V(G'_{i})|- 3 - 2}{2}\right\rfloor} = \\
    = k + |V(G'_1)| - 2 + |V(G'_{k+1})| - 2 + \sum_{i=2}^{k}{(|V(G'_{i})|-3)} =\\
    = k + |V(G'_1)| - 2 + 1 + \sum_{i=2}^{k+1}{(|V(G'_{i})|-3)} = \\
    = k + |V(G_1)| - 1 + \sum_{i=2}^{k+1}{(|V(G_{i})\setminus V(G_{i-1})| + 2 -3)} =\\
    =k + |V(G_1)| - 1 + \sum_{i=2}^{k+1}{(|V(G_{i})\setminus V(G_{i-1})|)}- k 
    = \\
    =|V(G_1)| + \sum_{i=2}^{k+1}{(|V(G_{i}) 
\setminus V(G_{i-1})|)} - 1 =|V(G_{k + 1})| - 1 = |V(G)| - 1.
\end{gather*}
\end{proof}		

Since all trees are outerplanar and interval colorable \cite{Kampreprint},
it is easy to see that for caterpillar trees $T_{n}$ on $n$ vertices, $W(T_{n})=n-1$ \cite{AsrDenHag}. In \cite{GiaroKubaleMalaf}, it was shown that fan graphs $F_{n}$ on $n$ vertices ($n>3$) have an interval $(n-1)$-coloring.
Since fan graphs $F_{n}$ are outerplanar, for any positive integer $n>3$, $W(F_{n})=n-1$. This shows that all proved upper bounds on $W(G)$ are sharp.

\begin{acknowledgement}
We would like to thank Petros A. Petrosyan for his suggestions and support.
\end{acknowledgement}

\bigskip

\textbf{Note added in proof} During the process of writing this paper we observed that L. Hollom, J. Portier and L. Versteegen submitted a preprint \cite{HollomPortierVersteegen} on ArXiv, where the authors independently proved Axenovich's conjecture. Our proofs are different, and our result on interval edge-colorings of outerplanar graphs is not covered by the result of L. Hollom, J. Portier and L. Versteegen. We also think that a proof technique in Theorem 2.4 can be useful for answering to Questions 5.1-2 \cite{HollomPortierVersteegen}.

\end{document}